\documentclass{sig-alternate}
\pdfoutput=1
\usepackage[T1]{fontenc}
\usepackage{txfonts}
\usepackage{microtype}
\usepackage[scaled=.84]{beramono}
\usepackage{relsize}
\usepackage{xspace}
\usepackage[normalem]{ulem}
\usepackage{flushend}
\usepackage{amsmath}
\usepackage{amssymb}
\usepackage{graphicx}
\usepackage[]{algorithm2e}
\usepackage{fixltx2e}
\usepackage{algorithmic}
\usepackage{tabulary}
\usepackage[utf8x]{inputenc}

\usepackage{cleveref}
\usepackage{wrapfig}

\newtheorem{mydef}{Definition}

\newtheorem{lemma}{Lemma}
\newtheorem{myproposition}{Proposition}

\usepackage{tikz}
\definecolor{lightblue}{rgb}{0,0,0}

\usepackage{mathtools}

\providecommand{\AJ}[2][]{}

\newcommand{\mli}[1]{\ensuremath{\mathit{#1}}}
\providecommand{\WorkItemSet}{\ensuremath{\mli{WorkItem}}\xspace}
\providecommand{\WorkItem}{\ensuremath{\mli{workitem}}\xspace}
\providecommand{\WorkItems}{\ensuremath{\mli{workitems}}\xspace}
\providecommand{\wis}{\ensuremath{\mli{wis}}\xspace}
\providecommand{\WIS}{\ensuremath{\mli{WIS}}\xspace}
\providecommand{\PF}{$\pi$\xspace}
\providecommand{\PFM}{\pi}
\providecommand{\secref}{ Section~\ref}
\providecommand{\defref}{ Definition~\ref}
\providecommand{\propref}{ Proposition~\ref}
\newcommand{\powerset}[1]{\mathbb{P}(#1)}

\def\sharedaffiliation{%
\end{tabular}
\begin{tabular}{c}}

\usepackage{float}

\floatstyle{ruled}
\newfloat{algorithm}{tbh}{lop}
\floatname{algorithm}{Algorithm}

\usepackage{caption}
\usepackage[subrefformat=parens,labelformat=parens]{subcaption}

\begin{document}

% Copyright
%\setcopyright{acmcopyright}
%\setcopyright{acmlicensed}
%\setcopyright{rightsretained}
%\setcopyright{usgov}
%\setcopyright{usgovmixed}
%\setcopyright{cagov}
%\setcopyright{cagovmixed}

% DOI
%\doi{10.475/123_4}

% ISBN
%\isbn{123-4567-24-567/08/06}

%Conference
%\conferenceinfo{SPAA '16}{June 27--29, 2016, Asilomar State Beach, CA, USA}

%\acmPrice{\$15.00}

%
% --- Author Metadata here ---
%\conferenceinfo{WOODSTOCK}{'97 El Paso, Texas USA}
%\CopyrightYear{2007} % Allows default copyright year (20XX) to be over-ridden - IF NEED BE.
%\crdata{0-12345-67-8/90/01}  % Allows default copyright data (0-89791-88-6/97/05) to be over-ridden - IF NEED BE.
% --- End of Author Metadata ---

\title{Abstract Graph Machine}
%\subtitle{[Extended Abstract]}
%\titlenote{A full version of this paper is available as
%\textit{Author's Guide to Preparing ACM SIG Proceedings Using
%\LaTeX$2_\epsilon$\ and BibTeX} at

\numberofauthors{4} %  in this sample file, there are a *total*
\author{
% 1st. author
\alignauthor Thejaka Amila Kanewala
% 5th. author
\alignauthor Marcin Zalewski
       %\affaddr{NASA Ames Research Center}\\
       %\affaddr{Moffett Field}\\
       %\affaddr{California 94035}\\
       %\email{fogartys@amesres.org}
%\alignauthor Martina Barnas\\
       %\affaddr{NASA Ames Research Center}\\
       %\affaddr{Moffett Field}\\
       %\affaddr{California 94035}\\
       %\email{fogartys@amesres.org}
% 6th. author
\alignauthor Andrew Lumsdaine\\
      \sharedaffiliation
      \affaddr{Center for Research in Extreme Scale Technologies (CREST)}\\
      \affaddr{Indiana University, IN, USA}\\
      \email{\{thejkane,zalewski,lums\}@indiana.edu}
}

%\affaddr{Center for Research in Extreme Scale Technologies (CREST)}\\
%\affaddr{Indiana University, IN, USA}\\
%\email{\{thejkane,zalewski,lums\}@indiana.edu}

% There's nothing stopping you putting the seventh, eighth, etc.
% author on the opening page (as the 'third row') but we ask,
% for aesthetic reasons that you place these 'additional authors'
% in the \additional authors block, viz.
%\additionalauthors{Additional authors: John Smith (The Th{\o}rv{\"a}ld Group,
%email: {\texttt{jsmith@affiliation.org}}) and Julius P.~Kumquat
%(The Kumquat Consortium, email: {\texttt{jpkumquat@consortium.net}}).}
\date{05 February 2016}
% Just remember to make sure that the TOTAL number of authors
% is the number that will appear on the first page PLUS the
% number that will appear in the \additionalauthors section.

\maketitle
\begin{abstract}
%A mathematical model for distributed memory parallel
%graph algorithms is presented.

An \textit{Abstract Graph Machine}(AGM) is an 
abstract model for distributed memory parallel
\textit{stabilizing} graph algorithms. A stabilizing
algorithm starts from a particular \textit{initial} state and
goes through series of different state changes
until it converges. 
The AGM adds work dependency to the stabilizing algorithm.
The work is processed within the \textit{processing
function}. All processes in the system
execute the same processing function.
Before feeding work into the processing
function, work is ordered using a \textit{strict weak ordering} relation.
The strict weak ordering relation divides work into
\textit{equivalence classes}, hence work within a single
equivalence class can be processed in parallel, but
work in different 
equivalence classes must be executed in the order
they appear in equivalence classes.
The paper presents the AGM model, semantics and AGM models for several
existing distributed memory parallel graph algorithms.
\end{abstract}


\section{Introduction}
\label{intro}
%topic : Graphs are ubiquitous
%point : Why graphs 
Graphs are ubiquitous data structures.
Many real-world relations are formulated
as graphs and graph algorithms are used to 
derive various characteristics related to those
relations. Applications such as social networks,
web search and scientific computing etc., 
use graph algorithms to derive useful information about graphs.
As for many other data, graph relations are
also growing so that they cannot be fit into
a graph data structure in a single process.
Those graphs need to be distributed
among several
computing resources and must be processed in
parallel.

%topic : distributed graph processing
%point : larger scale graphs must be processed in distributed memory

%topic : writing distributed memory algorithms is hard
%point : distribution, synchronization, message passing

Designing, implementing and analyzing distributed memory
parallel algorithms is inherently a difficult task
due to several reasons: 1. distributed memory
parallel algorithms depend on the data distribution
used, 2. designing and implementing proper mutual
exclusion and locking methods for distributed memory
systems is hard, 3. graphs are irregular in memory
access pattern. Therefore, the performance
of graph algorithms is not predictable as in 
for other regular algorithms, 4. due to irregularity,
graph algorithms heavily depend on the underlying
run-time and the architecture.

To provide better solutions to above
challenges, abstractions that help to understand the nature
of distributed memory parallel graph algorithms are
important. Developing such abstractions is difficult,
due to the discrepancy between the solution approaches
used in those distributed memory
parallel graph algorithms.
%Most of the distributed memory parallel graph algorithms
%are different from each other in their approach
%to solve the graph problem.
For example, \textit{Dijkstra's Single Source Shortest Path}~\cite{dijkstra1959note}
algorithm starts from a given source vertex and
spread its search through neighbors, but \textit{Bor\r{u}vka's Minimum Spanning Tree}~\cite{nevsetvril2001otakar} 
algorithm rely on set operations (disjoint union) to calculate the
Minimum Spanning Tree (MST). 
While Bor\r{u}vka's Minimum Spanning Tree (MST) uses set 
operations in its solution, the Dijkstra's Single Source Shortest Path (SSSP)
algorithm relies on priority based ordering of distances of neighbors
to calculate the SSSP.
Therefore, the solution approach used in Dijkstra's 
SSSP
is different from the solution approach used in Bor\r{u}vka's
MST algorithm.

\begin{figure}
\centering
\includegraphics[width=.75\linewidth]{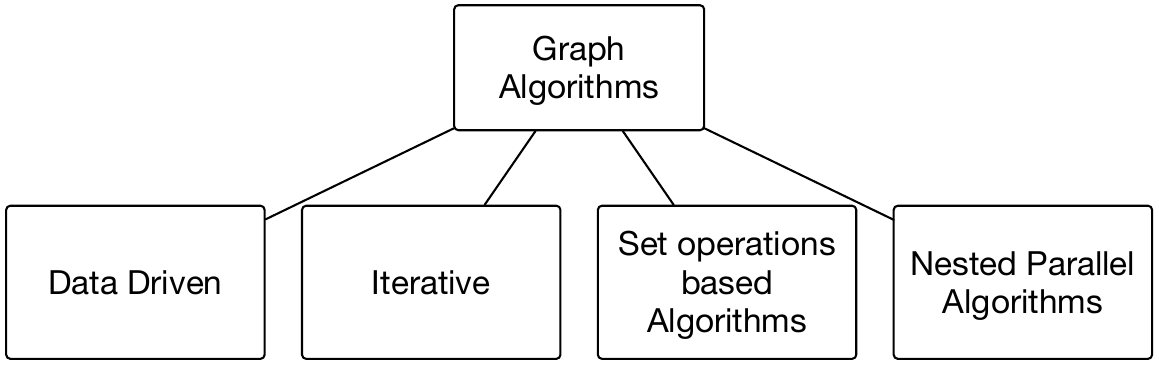}
\caption{Graph algorithm classification based on their solution approach.}
\vspace{-5ex}
\label{fig:classes}
\end{figure}

Based on the solution
approach, we classify parallel graph algorithms into four categories (Figure~\ref{fig:classes}):
1. Data-driven algorithms
2. Iterative algorithms
3. Set operations based algorithms
4. Nested parallel algorithms.

\textit{Data-driven algorithms} associate 
a \textit{state} to each vertex. At the start of the algorithm,
vertex states are initialized to a specific
value. The algorithm starts by changing the states of a subset of
vertices. Whenever a vertex state is changed, its neighbors
are notified. The state change is spread by notifying changes
to neighbors. Towards to end of the algorithm, state changes will 
be reduced and states reach a fixed point. When there is no
more state changes the algorithm terminates. An example of
a data-driven algorithm is the Dijkstra's SSSP.

As in data-driven algorithms,
\textit{iterative algorithms} also associate a state to a vertex, but 
instead of starting
from a subset of vertices, iterative algorithms
refine states associated with all vertices 
several iterations
until all vertex states satisfy a defined condition.
For example, in \textit{PageRank}~\cite{page1999pagerank},
the rank is calculated in iterations until rank values associated
with all vertices are less than the defined error value. \textit{Shiloach-Vishkin
Connected Components} (CC)~\cite{shiloach1982logn} is another
example of an iterative algorithm.

Part of the \textit{set operations based algorithms}, chase 
neighbors through edges, which is the standard
action, should take place in a graph algorithm and rest
of the algorithm rely on
set operations. An example is Bor\r{u}vka's
MST.
Bor\r{u}vka's MST uses disjoint sets of components and relies on
merging those components through set union operation to build
the MST. Another example of a parallel graph algorithm
that uses set operations is the 
\textit{Divide \& Conquer Strongly Connected Components} (DCSCC)~\cite{fleischer2000identifying}.  
In DCSCC, the vertex set is divided into two sets based on predecessor, 
successor reach-ability of a random vertex. Those two sets 
intersect to calculate the strongly connected component.

The \textit{nested parallel graph algorithms}
usually have an iterative or data-driven part. However,
enclosing the iterative or data-driven section these algorithms
have another iterative loop. An example for these kind of algorithms
is the \textit{Parallel Betweeness Centrality}~\cite{bader2006parallel} algorithm. Another
example is \textit{All-Pairs Shortest Path} algorithm.

Out of the algorithms discussed above, data-driven
algorithms and iterative algorithms can be formulated as
\textit{stabilizing} algorithms, in which
those algorithms start from a specific initial
vertex state and goes through several state
changes and converges. 
When all vertex states converged, the algorithm
terminates.
Further, those algorithms rely on standard graph
operations rather than set operations.
This paper proposes
an abstract model for converging graph algorithms
that only rely on graph operations.  

%For example Dijkstra's
%\textit{Single Source Shortest Path} (SSSP)~\cite{}
%algorithm starts from a given source vertex and
%spread its search through neighbors but Bor\r{u}vka's 
%\textit{Minimum Spanning Tree}~\cite{} 
%algorithm uses disjoint sets to merge components. 
%Though, some of the existing sequential and parallel graph
%algorithms that depend on set operations show low time complexity values 
%under sequential execution models and PRAM~\cite{} models, the
%performance of those algorithms suffers in 
%distributed execution mainly due to 
%global synchronization overheads.

%Even though, graph is a pointer
%based data structure, most of the graph algorithms use
%techniques based on set operations to achieve better complexity.

The \textit{Abstract Graph Machine} (AGM) is a 
mathematical abstraction for stabilizing graph algorithms.
In AGM each vertex is associated with a \textit{state}. 
States are changed by data propagating through edges. We call
data propagating through edges \textit{work}. The work
is executed on a uniform function, i.e. same
function is executed in every process.
We call this function \textit{processing function}.
The processing function takes a unit of work as the parameter
and may generate more work. Before processing, 
the work is ordered using
a \textit{strict weak ordering} relation. 
The strict weak ordering relation divides work into
\textit{equivalence classes}.

By dividing work into equivalence classes, the AGM
controls the rate at which algorithm converges.
When the amount of work in a single equivalence
class is higher the amount of available data parallelism
is also higher. However, when equivalence class is large
the amount of \textit{unnecessary work} (work that does
not contribute to the final algorithm state) is also higher.
The amount of
work in a single equivalence class is decided by the nature
of the strict weak ordering relation defined. 
The strict weak ordering relation also imposes an
induced ordering on the equivalence classes.
The equivalence classes are processed according to
the induced order. Work within a single equivalence
class can be processed in parallel but the processing of
work in
different equivalence classes is ordered.

The proposed abstraction provides a systematic
way to study stabilizing graph algorithms. Using AGM,
we can generalize existing data-driven graph algorithms
and also
derive new algorithms by composing orderings or
by introducing new attributes into the definition
of work. Further, AGM can be used for cost analysis
of distributed memory parallel graph algorithms.

In general, we will assume graph $G=(V,E)$,
where $V$ is the set of vertices and $E$ is the
set of edges. The algorithm states are maintained
in property maps and we assume graph and relevant
property maps are distributed  as a 1D distribution.
After defining necessary prerequisites, we give the
definition of an AGM in\secref{agm}.\secref{ddalgo}, models data-driven 
graph algorithms in AGM. We discuss AGM applicability 
to non-data-driven algorithms in\secref{agmnondd}.

%The AGM model can be implemented as a framework.
%The framework only take parameters to the AGM. The processing
%function is a stateless entity that is distributed
%among participating nodes. The definition of work
%correspond to the definition of the message. 

%In addition to the AGM model we also 
%present several examples of modeled graph algorithms.
%The modeled graph applications include
%SSSP, BFS, Connected Components, PageRank and Maximum Flow.
%In \textit{Abstract Graph Machine} (AGM) we rely
%on the fact that graph is a pointer based data-structure
%and graph's vertices are accessed through its neighbors.

%topic : AGM 
%point : AGM is a uniform way to express stabilizing graph algorithms

%topic : AGM covers only a class of algorithms
%point : not all algorithms cannot be expressed using AGMs

%topic : Uses of AGM
%point : Provides a way to derive new orderings, 

%topic : implementation friendly
%point : work as messages and processing functions as uniformly executing function

%topic : model to analyze cost
%point : cost analysis model

% define Graph
\AJ{TODO : Our work is focused on unstructured graphs}
\AJ{define what is a data-driven algorithm}
\AJ{discuss about other types of algorithms}
\AJ{what is the main point of the paper ?}
\AJ{data-driven or work driven ?}
\AJ{Work profile for each type of algorithm}
\AJ{Explain what is a data-driven algorithm}
%An \textit{Abstract Graph Machine} abstractly
%represents a distributed parallel data-driven graph algorithm as
%a function that can be executed in parallel
%on a single unit of work 
%and a \textit{strict weak ordering} impose on work.

\section{Abstract Graph Machine}
\label{agm}

In this section, we present the \textit{Abstract Graph Machine}.
First, we will define some of the terminologies
that we will be using. Then, we will
layout the structure of the AGM.

The main function that encapsulates the
logic of a stabilizing algorithm is called
the \textit{processing function}. Parameters
to the processing function is a single
unit of work and we call it a \WorkItem. The definition
of the \WorkItem depends on the state/s in which
algorithm is focused on and a \WorkItem
must be indexed with a vertex or
an edge.
The set of all the \WorkItems that algorithm
generates is the set 
\WorkItemSet. 
When the processing function, processes a \WorkItem,
it may change the states associated
with vertices.
For example, in SSSP, the state is the distance
from the source vertex and
the processing function resembles the logic
inside ``relax''.
For SSSP, a \WorkItem
consists of a vertex and  
the distance associated to
the vertex and $\WorkItemSet \subseteq (V \times \mathbb{R}_+^*)$.

An \textit{Abstract Graph Machine}(AGM) consists of a definition of a \textit{\WorkItemSet} set,
an \textit{initial \WorkItem set},
a set of \textit{states},
a \textit{processing function}
and a \textit{strict weak ordering} relation.
In the following subsections, we discuss each of these parameters 
in detail.

\subsection{The WorkItem Set}
\label{sec:workitems}

A \WorkItem is a tuple that has a vertex or an edge as of its first
element. If the first element of a \WorkItem is a vertex,
then we call that \WorkItem a \textit{vertex indexed \WorkItem}
and if the first element is an edge we call that \WorkItem
an \textit{edge indexed \WorkItem}.
The additional elements in the \WorkItem, carry state data local to the vertex or an edge.
For example, distance in a SSSP algorithm is an additional element
in the ``SSSP algorithm \WorkItem''. In addition to vertex (or edge)
and state data, a \WorkItem may carry \textit{ordering
attributes}. Ordering attributes are used when
defining the strict weak ordering relation.

A \WorkItem is constructed and consumed within a processing 
function. 
The processing function, that constructs the \WorkItem
may not reside in the same locality as the processing function, that consumes
the \WorkItem. In other words, \WorkItems can
travel from one locality to another. A \WorkItem
destination locality is decided based on the
data distribution. For a 1D distributed
graph and for a vertex indexed \WorkItem, the destination
locality is decided based on the ownership of 
the indexed vertex (i.e. the first element of the \WorkItem).
Details about data distribution is discussed in\secref{sec:datadis}.

All the \WorkItems generated by an algorithm is the set \WorkItemSet.
\WorkItemSet is
formally defined in\defref{def-graph-wi}.

\begin{mydef}\label{def-graph-wi}
  The \textbf{\WorkItemSet} is a set. For a given graph, G = (V, E), the 
  $\WorkItemSet \subseteq (V \times P_0 \times P_1 \dots \times P_n)$ where each $P_i$ represents a state value or an ordering attribute value.
  A \WorkItem $\in \WorkItemSet$ is represented as a tuple (e.g., \WorkItem = \textless$v, p_0, p_1 \dots, p_n$\textgreater $\:$ where $v \in V$ and each $p_i \in P_i$).
\end{mydef}

\AJ{TODO : Mention how work item tuple values are accessed. i.e. ``[]'' operator}
To access values in a \WorkItem tuple
AGM uses \textit{bracket operator}. e.g., if $w \in \WorkItemSet$ and
if $w$ = \textless$v, p_0, p_1\dots, p_n$\textgreater
 then $w[0]$ = v and $w[1]$ = $p_0$ and $w[2]$ = $p_1$,
etc.

\subsection{Data Distribution}
\label{sec:datadis}
Data distribution is implicit in AGM
and AGM always assumes that
each vertex (or an edge)
has a single \textit{owner} node (process).
Ownership of a \WorkItem is decided
by the ownership of the indexed
vertex (or edge) of the \WorkItem.
When a \WorkItem is being produced by a 
processing function it must be sent to
its appropriate owner node for execution.

In addition, states are also
distributed based on vertex (or edge)
distribution. State value for a
vertex (or edge) is only
maintained at the owner. 

%If a remote
%node changes a state value, change
%must be propagated \textit{synchronously}. 

\subsection{States}

AGM uses \textit{mappings} to represent states.
Each vertex or edge records
a value algorithm is calculating. Collectively, all the values
recorded against vertices or edges is treated as \textit{the} state
of the algorithm.
States are read and updated by the processing function.
In AGM terminology,
accessing a state value associated with a vertex (or edge) ``v''
is denoted as ``mapping\_name(v)'' (E.g :- distance(v), where distance
is a mapping from vertex to distance from source in SSSP).
\AJ{Unify this across the paper}

In addition to state mappings, processing
functions access read-only graph properties. For example,
``edge weight'' is read as a read only property. In terms of 
syntax, AGM does not distinguish between a read-only 
property map and a state mapping. However, read-only graph properties
such as ``edge weight'' are part of the graph definition.

States and read-only graph properties are only used 
within the processing function.
Further, local updates to states are made \textit{atomically}.

\subsection{Processing Function}

The \textit{processing function} (\PF) takes a \WorkItem as an argument
and may produce more \WorkItems (or 0) based on the logic defined inside
the \PF. Mathematically, \PF is declared as
$\mli{\PFM} : \WorkItemSet \longrightarrow \powerset{\WorkItemSet}$.

The processing function consists of a set of \textit{statements}
(\defref{pf}).
Each statement specifies, how the output \WorkItem (i.e., $w_{new}$)
should be constructed (\textit{<constructor>}), a condition based on
input \WorkItem and/or states (\textit{<condition>}) and an update 
to states (\textit{<state\_update>}).
A statement
is invoked  if the condition evaluated
to true. A statement may not generate new \WorkItems but only changes
a state based on a condition. If a statement is only making changes to
states based on input \WorkItems, then it produces $w_{nil}$ \WorkItem as the output.
The \WorkItem $w_{nil}$, is not treated as an active \WorkItem; i.e.,
$w_{nil}$ is not ordered and also $w_{nil}$ is not fed to any of the 
processing functions for further processing.

\begin{mydef}\label{pf}
  $\mli{\PFM} : \WorkItemSet \longrightarrow \powerset{\WorkItemSet}$
\begin{equation*}
  \mli{\PFM}(w) =  \begin{cases}
    \{ w_{new} | <constructor>, <state\_update>, \\
    \;\;\; <condition> \}
    \\
    \{ w_{new} | <constructor>, <state\_update>, \\
    \;\;\; <condition> \}
    \\
    \dots
    \\
    \{\} \;\;\;\; \mli{else}
  \end{cases}
\end{equation*}
\end{mydef}

An algorithm starts by invoking processing function
with the \textit{initial \WorkItem set}. Output \WorkItems of
the \PF are ordered according to the strict weak ordering
defined on \WorkItems. Ordered \WorkItems are then again
fed into the processing function. The interaction between
\PF and ordering is graphically depicted in Figure~\ref{fig:agm}.

\begin{figure}
\centering
\includegraphics[width=.55\linewidth]{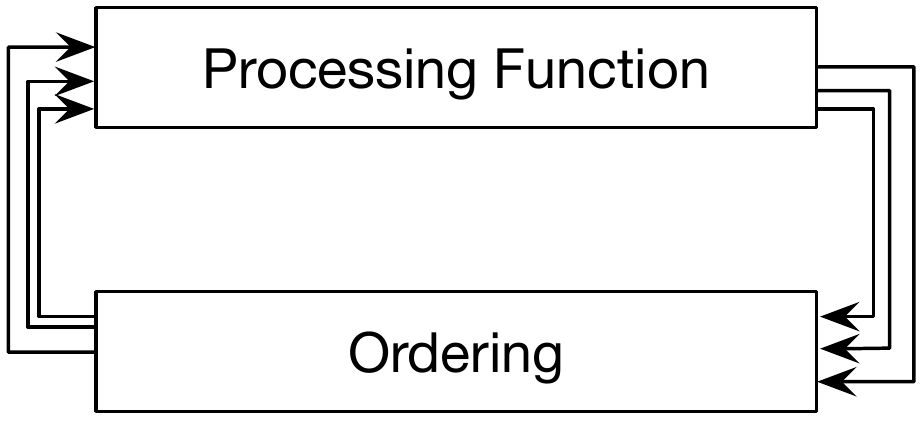}
\caption{Interaction between ordering and the processing function in an AGM}
\vspace{-1ex}
\label{fig:agm}
\end{figure}

\subsection{Ordering}

The AGM orders, output \WorkItems of the processing function
using a \textit{strict weak ordering} relation (denoted by $<_{\wis}$).
The strict weak ordering relation divides \WorkItemSet into 
equivalence classes based on the \textit{comparable}
relation ($<_{\wis}$).
The strict weak ordering relation
must satisfy following properties; \AJ{Or constraints ?}
\begin{enumerate}
\item For all $w \in \WorkItemSet$ $w \nless_{\wis} w$.
\item For all $w_1, w_2 \in \WorkItemSet$ if $w_1 <_{\wis} w_2$ then $w_2 \nless_{\wis} w_1$.
\item For all $w_1, w_2, w_3 \in \WorkItemSet$ if $w_1 <_{\wis} w_2$ and $w_2 <_{\wis} w_3$ then $w_1 <_{\wis} w_3$.
\item For all $w_1, w_2, w_3 \in \WorkItemSet$ if $w_1$ not comparable with $w_2$ and $w_2$ not comparable with $w_3$ then $w_1$ is not comparable with $w_3$.
\end{enumerate}

Properties 1 and 2 states that the strict
weak ordering relation is \textbf{not} \textit{reflexive} and
\textit{antisymmetric}. Property 3 denotes the
\textit{transitivity} of the ``comparable \WorkItems''
and Property 4 states that transitivity is preserved among
non-comparable elements in the \WorkItem set.
These properties give rise to an \textit{equivalence}
(i.e. non-comparable \WorkItems belong to the same
equivalence class)
relation defined on \WorkItem set, hence partition
the complete \WorkItem set. Since \WorkItems in
different equivalence classes are comparable, the
strict weak ordering relation defined on \WorkItem
set \textit{induces an ordering} on generated equivalence
classes. In general, there are several ways to define the induced
ordering relation (denoted $<_{\WIS}$), for our work we stick to the definition
given in~\defref{def-graph-aws-relation}.

%is there a better way to break lines in math mode ?
\begin{mydef}\label{def-graph-aws-relation}
  $<_{\WIS}$ is a binary relation defined on $\powerset{\WorkItemSet}$, such that if $W_1, W_2 \in
  \powerset{\WorkItemSet}$ then; $W_1 \le_{\WIS} W_2
  \; iff \\
  \; forall \; w_1 \in W_1 \; and \; forall \; w_2 \in W_2 ; w_1 <_{\wis} w_2$.
\end{mydef}

\begin{figure}%{.45\linewidth}
\centering
\includegraphics[width=.45\linewidth]{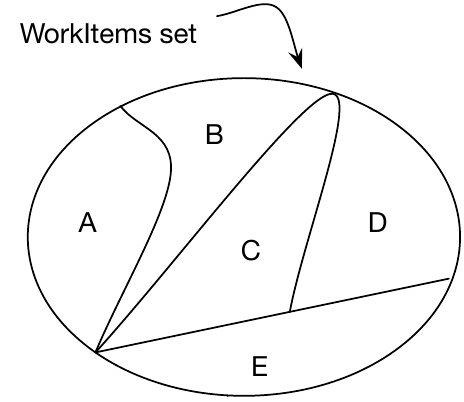}
\caption{Partitioning \WorkItems based on strict weak order relation $<_{\wis}$.}
\vspace{-1ex}
\label{fig:wipartition}
\end{figure}

Figure~\ref{fig:wipartition}, shows how \WorkItemSet is
partitioned by the strict weak ordering
relation $<_{\wis}$. Sets A, B, C, D, E are mutually exclusive
and $A \cup B \cup C \cup D \cup E = \WorkItemSet$.
Elements (\WorkItems) in A are not comparable using $<_{\wis}$, i.e. if
$w_1, w_2 \in A$ then $w_1 \nless_{\wis} w_2$ nor $w_2 \nless_{\wis} w_1$.
Same applies to other sets.
Also, induced relation order partitions in a
sequence. E.g., $B <_{\WIS} D <_{\WIS} A <_{\WIS} E <_{\WIS} C$;
as per this example the smallest partition defined by $<_{\WIS}$ is B.

\subsection{The AGM}
The AGM is formally defined in~\defref{def-graph-agm}.
\begin{mydef}\label{def-graph-agm}
  An \textbf{Abstract Graph Machine}(AGM) is a 6-tuple (G, \WorkItemSet, Q, \PF, $<_{\wis}$, S), where
  \begin{enumerate}
  \itemsep0em 
  \item G = (V, E) is the \textbf{input graph},
  \item $\textbf{\WorkItemSet} \subseteq (V \times P_0 \times P_1 \dots \times P_n)$ where each $P_i$ represents a state value or an ordering attribute,
  \item Q - Set of states represented as property maps,
  \item $\PFM : \WorkItemSet \longrightarrow \powerset{\WorkItemSet}$ is the \textbf{processing function},
  \item $<_{\wis}$ - \textbf{Strict weak ordering relation} defined on \WorkItems,
  \item S ($\subseteq \WorkItemSet$) - \textbf{Initial \WorkItem set}.
  \end{enumerate}
\end{mydef}

AGM execution starts with the \textit{initial \WorkItem set}.
The initial \WorkItem set is ordered according to the strict weak
ordering relation. Then \WorkItems within the smallest equivalence class
is fed to the \PF.
If \PF generates new \WorkItems,
they are again, separated into equivalence classes.
The \WorkItems within a single equivalence class can execute \PF
in parallel. However, \WorkItems in two different equivalence
classes must be ordered according to the induced relation (i.e. $<_{\WIS}$).
When executing \WorkItems in an equivalence class, it may generate
new \WorkItems to the same equivalence class or to an equivalence
class greater (as per $<_{\WIS}$)
than currently processing equivalence class.
The AGM executes \WorkItems in next equivalence class,
once it finished executing all the workitems in the
current equivalence class. An AGM terminates when it
executes all the \WorkItems in all the equivalence classes.

\section{Data-driven Algorithms in AGM}
\label{ddalgo}
A \textit{Data-driven algorithm} starts from a subset of \WorkItem set,
and generates more work as it progress. Towards to the end, the algorithm
generates less work and comes to a termination when it does not
generate more work. Most of the SSSP algorithms including Dijkstra's Algorithm,
$\Delta$-Stepping algorithm~\cite{meyer2003delta}, KLA~\cite{harshvardhan_kla:_2014} SSSP algorithm
are data-driven algorithms. Also, the \textit{Level Synchronous 
Breadth First Search}~\cite{bulucc2011parallel}, \textit{KLA Breadth First Search}~\cite{harshvardhan_kla:_2014}
and \textit{Push based Page Rank} discussed in ~\cite{whang2015scalable}
are also data-driven algorithms. 
In this section, we go through 
those data-driven algorithms and show the AGM formulation
for each.
Further, we present
a data-driven Connected Components algorithm
and an AGM formulation for that.
 
\subsection{Single Source Shortest Path Algorithms}
In this section, we go through several algorithms for SSSP
application and show how they can be modeled
using the Abstract Graph Machine defined in 
Definition~\ref{def-graph-agm}.
%Table~\ref{table:agmsummary}
%summarizes the Abstract Graph Machines modeled with their processing functions
%and strict weak orderings.

%\begin{table}
%  \centering
%  \begin{tabular}{l r r}
%    \toprule
%    AGM & Processing Function & Ordering \\ 
%    \midrule
%    Dijkstra's SSSP & $\mli{SSSP\_PF}$ & $<_{\mli{dj}}$ \\
%    $\Delta$-Stepping SSSP & $\mli{SSSP\_PF}$ & $<_\Delta$ \\
%    Chaotic SSSP & $\mli{SSSP\_PF}$ & $<_{\mli{chaotic}}$ \\
%    FIFO SSSP & $\mli{TIMED\_PF}$ & $<_t$ \\
%    LIFO SSSP & $\mli{TIMED\_PF}$ & $<_{\mli{t\_lifo}}$ \\
%    Level-Sync BFS & $\mli{BFS\_PF}$ & $<_{\mli{bfs}}$ \\
%    KLA BFS & $\mli{BFS\_PF}$ & $<_{\mli{kla}}$ \\
%    KLA SSSP & $\mli{KLA\_PF}$ & $<_{\mli{sssp\_kla}}$ \\ 
%    \bottomrule
%  \end{tabular}
%\caption{AGMs for algorithms discussed in Section~\ref{usecases} and their processing functions and orderings.}
%\label{table:agmsummary}
%\end{table}

In general, the \WorkItem set for SSSP application can be defined as
$\WorkItemSet^{sssp}  \subseteq (V \times Distance)$,
where \textit{Distance} $\subseteq \mathbb{R}_+^*$. SSSP algorithms
use \textit{distance} as the output state and \textit{weight} map as a 
read-only input mapping.
The processing function for SSSP changes the distance state
if input \WorkItem's distance is less than what is already stored in the
distance map. Further, adjacent vertices of a given vertex are
accessed through the \textit{neighbors} function.
(Declared as $neighbors: V \longrightarrow \powerset{V}$).

Interestingly, most of the SSSP algorithms share almost the same processing function.
In general, the SSSP processing function ($\PFM^{sssp}$) can be defined as follows;

\begin{mydef}\label{sssp-pf}
  $\PFM^{sssp} : \WorkItemSet^{sssp} \; \longrightarrow \WorkItemSet^{sssp}$
\begin{equation*}
  \PFM^{sssp}(w) =  \begin{cases}
    \{ w_k | < w_k[0] \in \mli{neighbors}(w[0]) \; \mli{and}
    \\ \;\;\; w_k[1] \longleftarrow w[1] + \mli{weight}(w[0], w_k[0]) >
    \\ \;\;\;< distance(w[0]) \longleftarrow w[1] >
    \\ \;\;\;< \mli{if} \; w[1] < \mli{distance}(w[0]) > \} 
    \\
    \{\} \;\;\;\; \mli{else}
  \end{cases}
\end{equation*}
\end{mydef}

The processing function definition given in\defref{sssp-pf}
is organized according to the general processing function
definition given in~\ref{pf}. The processing function, \PF
has two statements. The first statement is executed
only if \WorkItem distance is less than the value stored 
in the $distance$ state for the relevant vertex in \WorkItem. 
Constructor of the first statement
specifies how to construct a new \WorkItem. 
In\defref{sssp-pf}, $w$ refers to currently processing
\WorkItem and $w_k$ is the new \WorkItem that will be
constructed. Further, the bracket operator is used
to access \WorkItem elements (As discussed in\secref{sec:workitems}).
For $w \in \WorkItemSet^{sssp}$, the term $w[0]$
refers to a vertex and $w[1]$ refers to 
the distance associated with the vertex referred by $w[0]$.

\subsubsection{Dijkstra's Algorithm}

Dijkstra's SSSP algorithm is the work efficient SSSP algorithm.
Algorithm globally orders vertices by their associated
distances and shortest distance vertices are processed first.
In the following we define the ordering relation for Dijkstra's algorithm and,
we instantiate Dijkstra's algorithm using an AGM.

\begin{mydef}\label{dj-swo}
  $<_{\mli{dj}}$ is a binary relation defined on $\WorkItemSet^{sssp}$ as follows; 
  Let $w_1, w_2 \in \WorkItemSet^{sssp}$, then;
  $w_1 <_{\mli{dj}} w_2$ iff $w_1[1] < w_2[1]$
\end{mydef}

It can be proved that $<_{dj}$ is a strict weak ordering relation that
satisfies constraints listed under Definition~\ref{dj-swo} (proof is omitted).
AGM instantiation for Dijkstra's algorithm is given in~\propref{dijkstra-agm}.

\begin{myproposition}\label{dijkstra-agm}
  Dijkstra's Algorithm is an instance of an AGM where;
 \begin{enumerate}
 \itemsep0em 
 \item $G = (V, E, weight)$ is the input graph,
 \item \WorkItemSet = $\WorkItemSet^{sssp}$,
 \item Q = \{distance\} is the state mapping,
 \item \PF = $\PFM^{sssp}$,
 \item Strict weak ordering relation $<_{\wis}$ = $<_{\mli{dj}}$,
 \item S = \{\textless$v_s$, 0\textgreater\} where $v_s \in V$ and $v_s$ is the source vertex.
 \end{enumerate}
\end{myproposition}

\subsubsection{$\Delta$-Stepping Algorithm}

$\Delta$-Stepping~\cite{meyer2003delta} arrange vertex-distance pairs
into distance ranges (\textit{buckets}) of size $\Delta (\in \mathbb{N})$
and executes buckets in order.  Within a bucket, vertex-distance pairs
are not ordered,
and can be executed in any order.
Processing a bucket may produce extra work for
the same bucket or for a successive bucket.
The strict weak ordering relation for the $\Delta$-Stepping algorithm
is given in Definition~\ref{delta-swo}.

\begin{mydef}\label{delta-swo}
  $<_\Delta$ is a binary relation defined on $\WorkItemSet^{sssp}$ as follows;
  Let $w_1, w_2 \in \WorkItemSet^{sssp}$, then; \newline
  $w_1 <_{\Delta} w_2$ iff $\lfloor w_1[1]/\Delta \rfloor < \lfloor  w_2[1]/\Delta \rfloor$
\end{mydef}

Instantiation of the $\Delta$-Stepping algorithm in AGM (given in Proposition~\ref{delta-agm})
is same as in Proposition~\ref{dijkstra-agm},
except the strict weak ordering relation is $<_{\Delta}$ (= $<_{\wis}$).

\AJ{If we need more space we can remove redundant text}
\begin{myproposition}\label{delta-agm}
  $\Delta$-Stepping Algorithm is an instance of AGM where;
 \begin{enumerate}
  \itemsep0em 
  \item G = (V, E, weight) is the input graph
  \item \WorkItemSet = $\WorkItemSet^{sssp}$
  \item Q = \{distance\} is the state mapping,
  \item \PF = $\PFM^{sssp}$,
  \item Strict weak ordering relation $<_{\wis}$ = $<_{\Delta}$
  \item S = \{\textless$v_s$, 0\textgreater\} where $v_s \in V$ and $v_s$ is the source vertex.
 \end{enumerate}
\end{myproposition}

%\subsubsection{Chaotic SSSP Algorithm}
%\label{chaotic}
%In \textit{Chaotic SSSP algorithm} the \WorkItems are processed without an order. In this algorithm each \WorkItem
%is related to all other \WorkItems, in which we have a
%single equivalence class and \WorkItems in the single
%equivalence class is processed in any order.
%The strict weak order relation $<_{chaotic}$ is trivial for
%Chaotic SSSP, i.e. $w_1 <_{chaotic} w_2$ \textbf{False} for all $w_1, w_2 \in \WorkItemSet^{sssp}$.
%That is, none of the work items are related by $<_{\mli{chaotic}}$.

%The AGM instance for Chaotic Algorithm is same as Proposition~\ref{delta-agm}
%except the strict weak ordering relation.
%The ordering relation for Chaotic Algorithm should be $<_{\mli{chaotic}}$.

\subsection{Breadth First Search (BFS) Algorithms}

In this subsection, we present AGM models for 
two distributed memory parallel BFS algorithms.
They are, level synchronous BFS and KLA BFS.

\subsubsection{Level Synchronous BFS Algorithm}

The level-synchronous breadth first search algorithm
uses data structures to store the \textit{current} and \textit{next}
vertex frontiers.
Then the next container data is swapped with current
after processing each level.

In the following we model level-synchronous BFS with the AGM. 
The level synchronous BFS order work by the level in the resulting
BFS tree.
Therefore, the ordering attribute, that we are interested in, is the \textit{level};
hence we define $\WorkItemSet^{bfs} \subseteq V \times Level$
where $Level \subseteq \mathbb{N}$.

The processing function for BFS is defined in Definition~\ref{bfs-pf}.
The state of the BFS algorithm is maintained in a map
($vertex\_level : V \longrightarrow Level$)
that store the level associated with each vertex. An infinite value
(very large value) is associated to each vertex at the start of
the algorithm. Then the infinite value is changed as the algorithm
traverse through graph level by level. 

\begin{mydef}\label{bfs-pf}
  $\PFM^{bfs} : \WorkItemSet^{bfs} \longrightarrow \powerset{\WorkItemSet^{bfs}}$
\begin{equation*}
  \PFM^{bfs}(w) =  \begin{cases}
    \{ w_k | < w_k[0] \in neighbors(w[0]) \; and
    \\ \;\;\; w_k[1] \longleftarrow w[1] + 1 >
    \\ \;\;\; < vertex\_level(w_k[0]) \longleftarrow w[1] >
    \\ \;\;\;\;\;< (if \; vertex\_level(w_k[0]) < \infty) > \}
    \\ \{\} \;\;\;\; else
  \end{cases}
\end{equation*}
\end{mydef}

The AGM is instantiated for level-synchronous BFS algorithm
is given in\propref{bfs-agm}.

\begin{myproposition}\label{bfs-agm}
  Level Synchronous BFS Algorithm is an instance of an AGM where;
 \begin{enumerate}
  \itemsep0em 
  \item G = (V, E) is the input graph,
  \item \WorkItemSet = $\WorkItemSet^{bfs}$,
  \item Q = \{ vertex\_level \} is the state mapping,
  \item \PF = $\PFM^{bfs}$,
  \item The strict weak ordering relation $<_{\wis}$ = $<_{bfs}$,
  \item S = \{\textless$v_s$, 0\textgreater\} where $v_s \in V$ and $v_s$ is the source vertex.
 \end{enumerate}
\end{myproposition}

\subsubsection{KLA BFS Algorithm}
The KLA BFS is similar to the level-synchronous BFS discussed above. 
Unlike in level-synchronous BFS, 
the KLA BFS processes \WorkItems asynchronously up to $k$ levels. 
In other words, \WorkItems are partitioned based on the value of $k$.
The strict weak ordering relation for KLA is defined in~\defref{kla-rel}.

\begin{mydef}\label{kla-rel}
  $<_{kla}$ is a binary relation defined on $\WorkItemSet^{bfs}$ as follows:
  Let $w_1, w_2 \in \WorkItemSet^{bfs}$, then; 
  $w_1 <_{kla} w_2$ iff $\lfloor w_1[1]/k \rfloor < \lfloor w_2[1]/k \rfloor$
\end{mydef}

The AGM instantiation for KLA BFS is same as 
level synchronous BFS instantiation
(Proposition~\ref{bfs-agm}), except the ordering relation is replaced with $<_{kla}$.

\subsection{PageRank}
\textit{PageRank}(PR)~\cite{page1999pagerank} is a 
graph algorithm extensively used in web mining.
Given a graph $G=(V,E)$, the PageRank, $PR(v)$ of a 
vertex $v$ is calculated using the formula given in Equation 1.
The variable $\alpha$ represents the teleportation parameter.
Function \textit{source} returns the source vertex given an edge and
functions \textit{in\_edges} and \textit{out\_edges} respectively return in and out
edges of a given vertex.

\begin{equation}
\label{pr-formula}
PR(v) = (1-\alpha) + \alpha\sum_{e \in in\_edges(v)}\frac{PR(source(e))}{|out\_edges(source(e))|}
\end{equation}

In PageRank, web pages are modeled as vertices and links between
web pages are edges. The PageRank algorithm calculates a numeric
weight for each page, which describes the importance of a web page.

Often PageRank algorithm is implemented as an
\textit{iterative algorithm}; i.e. algorithm iterate through
all the vertices and calculates PageRank using the formula
given in Equation 1.
The algorithm continues to calculate rank values
until the different between newly calculated value and
the previous value is less than, a given error vale - $\epsilon$.

In the data-driven form of the algorithm, the PageRank
of a vertex depends on the neighbours connected to the vertex
using an in-edge. Whenever PageRank value of a neighbour connected
through an in-edge changes, the PageRank of the current 
vertex must be re-calculated. A PageRank algorithm that
uses this argument is also explained in~\cite{whang2015scalable}. 
The AGM formulation of the PageRank algorithm uses dependency between
vertices (through in-edges) in terms of PageRank calculation to generate
work.

In PageRank the final state we are interested in
is the rank values of  vertices.
A straightforward way to model work for PageRank is to
use vertex and rank value.
One way to reduce the amount
work is to
order \WorkItems  
by the \textit{residual} of a PageRank calculated
for a vertex. (Residual based ordering is discussed in~\cite{whang2015scalable}). \AJ{TODO : Verify !}.
Residual is the portion of the PageRank value that is being
pushed through an out edge of a vertex.
\AJ{Validate with Marcin}

Based on the residue value,
the \WorkItem set for PageRank can be defined
as $\WorkItemSet^{pr} \subseteq (V \times \mathbb{R})$,
where $\mathbb{R}$ is used to represent 
the residual value.
The processing
function for PageRank takes a \WorkItem ($\in \WorkItemSet^{pr}$)
and produces more \WorkItems if the difference
between newly calculated PageRank value and previous
PageRank value is greater than $\epsilon$. Further, the algorithm
uses $rank$ mapping to store the calculated PageRank values.
The processing function for PageRank is defined in\defref{pf-pr}.

\begin{mydef}\label{pf-pr}
  $\mli{\PFM^{pr}} : \WorkItemSet^{pr} \longrightarrow \powerset{\WorkItemSet^{pr}}$
\begin{equation*}
  \mli{\PFM^{pr}}(w) =  \begin{cases}
    \{ w_n | < w_n[0] \in out\_neighbours(w[0]) \; \& \\
    \;\;\; w_n[1] \leftarrow \delta > \\
    < pr_{new} \leftarrow PR(w[0]) \; ; \; \delta \leftarrow (pr_{new} - rank(w[0])) \\
    \;\;\; ; rank(w[0]) \leftarrow pr_{new} > \\
    < (\mli{if} \; \delta > \epsilon ) > \}
    \\
    \{\} \;\;\;\; \mli{else}
  \end{cases}
\end{equation*}
\end{mydef}

The PageRank algorithm converges quickly if we process
higher residue \WorkItems first. Therefore, there are 
several possibilities to define ordering for PageRank:
1. do a strict comparison on residue value,
2. define strict weak ordering as in $\Delta$-Stepping
3. do not perform ordering at all. Each of the orderings creates a 
different size of equivalence class on PageRank
\WorkItems.
However, for the formulation presented above, we
define strict weak ordering (\defref{s-w-o-pr}) only based on the comparison of
residue
values.
Other orderings are also possible.

\begin{mydef}\label{s-w-o-pr}
  $<_{pr}$ is a binary relation defined on $\WorkItemSet^{pr}$ as follows;
  Let $w_1, w_2 \in \WorkItemSet^{pr}$, then; \newline
  $w_1 <_{pr} w_2$ iff $w_1[1] >  w_2[1] $.
\end{mydef}

The PageRank algorithm starts by assigning an initial
rank to every vertex. Therefore the initial \WorkItemSet
has a \WorkItem per each vertex and a associated
residue value initialized to 0. More formally
the initial \WorkItemSet $IW^{pr} = \{w | w[0] \in V \; and \; w[1] \leftarrow 0\}$.

With necessary parameters in hand we define the AGM
formulation for PageRank algorithm in\propref{pr-agm}.

\begin{myproposition}\label{pr-agm}
  PageRank Algorithm is an instance of an AGM where;
 \begin{enumerate}
 \itemsep0em 
 \item $G = (V, E)$ is the input graph,
 \item \WorkItemSet = $\WorkItemSet^{pr}$,
 \item Q = \{rank\} is the state mapping,
 \item \PF = $\PFM^{pr}$,
 \item Strict weak ordering relation $<_{\wis}$ = $<_{\mli{pr}}$,
 \item S = $IW^{pr}$
 \end{enumerate}
\end{myproposition}

\AJ{State update should come before condition - change everywhere}
\subsection{Connected Components Algorithm}
For a given undirected graph, $G = (V, E)$, the \textit{connected component} (CC) 
is a \textit{subgraph}
in which, any two vertices are connected through a \textit{path}. A connected component
can also be defined as a \textit{reachable} relation. A vertex $v \; (\in V)$ 
is reachable to vertex $u \; (\in V)$ if and only
if there is a path from $v$ to $u$.

In the literature, we find two main types of parallel connected component algorithms:
1. \textit{Shiloach-Vishkin's $O(log n)$ based algorithms}~\cite{shiloach1982logn}, 
2. \textit{Search based algorithms}~\cite{hirschberg1979computing}.
Algorithms derived from Shiloach-Vishkin's connected components
are not data-driven algorithms. They are 
\textit{iterative algorithms} and will be discussed in 
Section~\ref{secite}. 
In the following we discuss a search based connected component
algorithm.

Search based algorithms mainly 
use \textit{Depth-first search} (DFS), \textit{Breadth-first search} (BFS)
or a \textit{Chaotic search} algorithm (See~\cite{hirschberg1979computing}). 
DFS based algorithms give little support to explore
the available parallelism. 
Both BFS based algorithms and Chaotic
search based algorithms have more opportunity
to explore parallelism. To present the AGM formulation
we will use a Chaotic search based algorithm.

A parallel search (Chaotic) based algorithm is given in Listing~\ref{search-cc}. 
Algorithm uses
a property map (\textit{component}) to record the component id of each vertex. 
Component map is the output state of the algorithm.
Initially, the value of component id of a vertex is assigned
to a very large number, then component map is 
updated when the function ``CC'' is invoked. The ``VertexComponent''
structure contains the vertex id and the component id for the vertex.

\begin{algorithm}
\algsetup{linenosize=\scriptsize}
%\small
\caption{Search Based CC Algorithm}
\label{search-cc}
\textbf{Input:} Graph $G=(V, E)$
\begin{algorithmic}[1]
  \STATE Initialize() \{
  \FOR{each vertex $v$ in $V$}
  \STATE $component(v) \longleftarrow \infty$
  \ENDFOR
  \STATE \}
  \STATE
  \STATE CC(vc : VertexComponent) \{
  \STATE $old\_component \longleftarrow component(vc.vertex)$
  \IF {$vc.component < old\_component$}
  \STATE $component(v) \longleftarrow vc.component$
  \FOR {each $v$ in adjacencies(vc.vertex)}
  \STATE CC(VertexComponent(v, vc.component))
  \ENDFOR
  \ENDIF
  \STATE \}
  \STATE
  \STATE Invoke() \{
  \STATE Initialize()
  \FOR{each vertex $v$ in $V$}
  \STATE CC(VertexComponent(v, v))
  \ENDFOR
  \STATE \}
\end{algorithmic}
\end{algorithm}

We can develop several versions of the above algorithm. The algorithm in Listing~\ref{search-cc}
is a chaotic search algorithm. We also can develop the Dijkstra's version
and $\Delta$-Stepping version of the algorithm.
In the following we model above algorithm using an AGM.

The \WorkItems for search based CC should include the distance
and the \textit{component id} (Since component id represents the state). 
Therefore, we define $\WorkItemSet^{cc} \subseteq (V \times Component)$.
In the definition the $Component \subseteq \mathbb{N}$.
The processing function for search based CC is similar to SSSP processing
function except it updates component id instead of the distance.
Processing function for CC is given in Definition~\ref{sb-cc-pf}.

\begin{mydef}\label{sb-cc-pf}
  $\PFM^{cc} : \WorkItemSet^{cc} \longrightarrow \powerset{\WorkItemSet^{cc}}$
\begin{equation*}
 \PFM^{cc}(w) =  \begin{cases}
    \{ w_k | < w_k[0] \in \mli{neighbors}(w[0]) \\
    \;\;\;\; and \; w_k[1] \longleftarrow  w[1] > \\
    \;\;\; < component(w[0]) \longleftarrow w[1] > \\
    \;\;\; < \mli{if} \; (w[1] < component(w[0]) > \} 
    \\
    \\
    \{\} \;\;\;\; \mli{else}
  \end{cases}
\end{equation*}
\end{mydef}

We order \WorkItems by component id so that the smallest component ids
are processed  first. The order \WorkItems processed
does not affect the correctness of the Algorithm~\ref{search-cc}.
Therefore, more relaxed ordering can also be applied
to CC AGM formulation. For simplicity we use
the ordering based on component id. The strict weak ordering
relation for CC is defined in Definition~\ref{sb-cc}.

\begin{mydef}\label{sb-cc}
  $<_{cc}$ is a binary relation defined on $\WorkItemSet^{cc}$ as follows;
  Let $w_1, w_2 \in \WorkItemSet^{cc}$, then; \newline
  $w_1 <_{cc} w_2$ iff $w_1[1] < w_2[1]$.
\end{mydef}

In Proposition~\ref{sb-cc-agm} we define the AGM for search based connected
components.
\begin{myproposition}\label{sb-cc-agm}
  CC Algorithm is an instance of an AGM where;
 \begin{enumerate}
  \itemsep0em 
  \item $G = (V, E)$ is the input graph
  \item \WorkItemSet = $\WorkItemSet^{cc}$
  \item Q = \{components\} is the state mapping,
  \item \PF = $\PFM^{cc}$
  \item Strict weak ordering relation $<_{\wis}$ = $<_{cc}$
  \item S = $\{w | w[0] \in V \; and \; w[1] = w[0]\}$.
 \end{enumerate}
\end{myproposition}

\section{AGM for Non-Data-driven \\ Algorithms}
\label{agmnondd}
So far we modeled data-driven algorithms in
AGM.
A natural question to ask is whether AGM
approach can be used to model other kinds
of graph algorithms. The answer depends on the type
of the graph algorithm we are focused on.
For some types of graph algorithms AGM
approach can be used with few modifications.
Some other graph algorithms cannot be modeled using
AGM due to the in-adequateness to express orderings
using strict weak orderings and use set operations
in those algorithms. In this section,
we analyze the applicability of AGM to several 
non data driven
algorithm types.

\subsubsection{Iterative Algorithms}
\label{secite}

An \textit{iterative algorithm} travels through
all the vertices (or edges) until algorithm
meets a specific state condition. Example
algorithms are \textit{Iterative PageRank}~\cite{gleich2004fast}
and \textit{Shiloach-Vishkin Connected Components}~\cite{shiloach1982logn}.
We use iterative PageRank as an example to dicuss
the applicability of AGM to iterative algorithms,
but discussion is general in which it applies
to other iterative algorithms also.

\textit{Iterative PageRank} iterate through
all the vertices until all vertices reach
a saturated PageRank value. PageRank value of
a vertex is saturated if the difference between
newly calculated PageRank value and the previous
PageRank value is less than a defined error 
value. In the parallel iterative PageRank vertices
in a single iteration are processed parallely and
after each iteration algorithm checks whether
the PageRank values has reached a saturation.

Vertices processed in parallel within
a single iteration
goes into a single equivalence class in the
AGM formulation of the iterative PageRank. But 
vertices belonging to different iterations should be placed
in different equivalence classes. Yet, after
processing each equivalence class AGM must check
whether PageRank has reached its saturation
state. 
\AJ{Maybe a figure}

The current, AGM formulation does not have
fascility to check a condition before processing
next equivalence class. The AGM can be
augmented to check a condition
after processing each equivalence class.
Since equivalence classes are defined
based on the iteration, a \WorkItem
of an iterative algorithm must have 
the \textit{iteration number}
as a member and strict weak ordering
must be defined in such a way, \WorkItems
that has same iteration numbers are 
not comparable.

%If algorithm needs further ordering
%that must be defined within the 
%equivalence class defined by 
%the iteration number. That is the
%equivalence classes generated by further
%ordering must be subsets of 
%equivalence classes generated by ordering
%\WorkItems according to iteration
%number. The advantage of further
%ordering for iterative algorithms
%may not be same as for data-driven
%algorithms since the equivalence
%classes are bounded and a condition
%needs to be checked after processing
%each equivalence class, yet an AGM
%augmented with a condition check
%can be used to express iterative
%graph algorithms.

\subsubsection{Divide \& Conquer Algorithms}

A \textit{Divide \& Conquer} algorithm recursively
breaks down a problem into two or more
sub problems, until they become sufficiently
simple to solve. In a parallel Divide \& Conquer
algorithm divisions are conquered parallely.
An example is \textit{Divide \& Conquer Strongly Connected Components} 
(DCSCC)~\cite{fleischer2000identifying}.
The question is whether we can use AGM
to express Divide \& Conquer algorithms.

The answer is \textit{no}, we cannot use AGM
to express parallel \textit{Divide \& Conquer}
algorithms. The main reason is that the strict
weak ordering cannot be used to express
the ordering in a Divide \& Conquer algorithm
in such a way it produces full available
data parallelism.
\AJ{review}
Ordering in divide \& conqer algorithms 
is organized based on subset relationships
of divisions. Therefore, the strict weak
ordering is too strong to represent
such ordering. In Appendix~\ref{appscc} we use DCSCC
 algorithm as an example to show
that divide \& conqer algorithm ordering
cannot be expressed using a strict
weak ordering. Rather, \WorkItems in DCSCC algorithm is
ordered using a \textit{partial order} (based on the subset relationship of
divisions).

\subsubsection{Algorithms with Nested Parallelism}

Some algorithms show multiple levels
of parallelism. Usually these algorithms
have a visible data parallelism then a second
level of \textit{task parallelism} or
\textit{data parallelism}.
Such nested parallel constructs are
not easily transferrable to processing
function so that it is amenable to
execute in a distributed environment.
We see these types of algorithms in
applications such as \textit{Minimum Spanning Tree} (MST)
and \textit{Betweeness Centrality}.
In the following we use MST as an example
and briefly analyze 
nested parallel algorithms in related
to AGM.

A \textit{spanning tree}, T of a graph G is a subgraph 
that includes all the vertices of G and is a tree. 
A minimum spanning tree (MST) of an undirected, 
connected, weighted graph G is a 
spanning tree that connects all the 
vertices with minimum weighted edges.
There are
3 widely used algorithms to solve MST problem. They
are 1. Prim's Algorithm~\cite{Cormen:2001}, 2. Kruskal Algorithm~\cite{Cormen:2001}
and 3. Bor\r{u}vka's Algorithm~\cite{nevsetvril2001otakar}. 

Both Prim's algorithm 
and Kruskal algorithm does not provide much data parallelism. 
However, Bor\r{u}vka's algorithm provides more
parallelism and most of the existing
parallel (and also distributed)
implementations are based on Bor\r{u}vka's
algorithm (See ~\cite{gallager1983distributed}.~\cite{garay1998sublinear}).
Bor\r{u}vka's algorithm starts with a forest
(i.e. each vertex is a component initially)
and finds the minimum 
weight edge betwee two components. 
If algorithm finds a such edge, 
components are connected using the found edge.

To find the representative component set 
Bor\r{u}vka's algorithm uses
\textit{disjoint union} data structure.
In a distributed setting calculating the
representative set for a given source vertex
and for a given target vertex of an edge can
be performed in parallel.
Therefore in addition to parallel
processing of components algorithm
also can process calculation of 
representative components for source
vertex and target vertex in parallel.
\AJ{TODO: Need a figure}

The current AGM formulation does not have
fascility to model nested levels of parallel work;
i.e. work for parallel component processing
and work for parallely calculate representative
vertex sets for source and target of an
edge. Another algorithm that shows the same behaviour
is \textit{Parallel Betweeness Centrality}~\cite{bader2006parallel}
algorithm. Both algorithms generate parallel work
in a \textit{fork-join} like structure.

\AJ{In future work mention that we expect to extend current AGM
to express nested levels of work}

\section{Conclusion}

In this paper we presented Abstract Graph Machine (AGM),
a mathematical model for distributed memory parallel
graph algorithms that converges. The model formulate 
edge traversals as work and express an algorithm
using a common function executed by every process in the system
and an ordering that divides work into equivalence classes.
We showed that existing data driven algorithms can be
modeled using the AGM and also iterative algorithms. Algorithms 
that converge using set operations cannot be modeled in the AGM.
We believe algorithms can be derived in such a way they can
be modeled in the AGM.

Modeling algorithms using AGM generalizes existing algorithms, 
also 
model allows us to derive new variations of algorithms
by changing the way algorithm order work. Ordering
of an algorithm can be changed by either changing the average
size of an equivalence class generated by the ordering or by
introducing new ordering attributes.
Further, in future we plan to
use AGM model to build cost models for distributed memory
parallel
graph algorithms. 

\section{Acknowledgments}
Authors (Thejaka) would like to thank
Prabath Silva, for his valuable input
on orderings in divided and conquer algorithms.
I (Thejaka) thank Martina Barnas for many discussions
about AGM algorithm abstractions and for her valuable feedback on
this paper.
Further, 
this material is based upon work 
supported by the National Science Foundation under Grant No. 1319520
and National Science Foundation under Grant No. 1111888.

\bibliographystyle{abbrv}
\bibliography{paper}  % sigproc.bib is the name of the Bibliography in this case
% ACM needs 'a single self-contained file'!
%\balancecolumns
%\newpage
\appendix
%Appendix A
\section{Strongly Connected \\ Components}
\label{appscc}
\textit{Strongly Connected Component}(SCC) is a subgraph of a directed
graph where every vertex is reachable from every other vertex in the
subgraph. 
The DCSCC algorithm for SCC selects a random vertex (\textit{pivot})
and divides vertices set as
vertices reachable from the pivot (\textit{descendents})
and vertices
that can reach the pivot (\textit{predecessors}).
~\cite{fleischer2000identifying} proves that intersetion of predecessors
and descendents is a strongly connected component
containing the pivot.
~\cite{fleischer2000identifying} also proves that other SCC are either in
predecessor set or descendent set or in the \textit{remainder}
($= V - (predecessors \cup descendents)$). Then, DCSCC
algorithm applies same procedure to descendents, predecessors
and to remainder.
Descendents (Let's call this \textit{FWD} set) and predecessors(\textit{BWD} set)
can be calculated in parallel.
Then, DCSCC algorithm finds a strongly connected component (\textit{SCC})
by calculating the set intersection of FWD and BWD. Afterwards, DCSCC algorithm
divides vertex set into 3 segments. Segment1 = $FWD-SCC$,
Segment2 = $BWD-SCC$ and Segment3 = $remainder$ (REM).
Then each segment is processed in parallel. We can express the parallel
execution of DCSCC in a tree as depicted in Figure~\ref{fig:dcscc}.

\begin{lemma}
 DCSCC algorithm \textbf{cannot} be modeled with an AGM.
\end{lemma}

\begin{proof}
  Proof is by contradiction.
  Suppose DCSCC can be expressed using an AGM.
  Then, WorkItems in DCSCC can be ordered using a strict weak ordering relation $R$.
  R induces an equivalence relation and the induced relation partition work so that,
  work items in the same partition can be executed in parallel. Consider Figure~\ref{fig:dcscc}.
  We should be able to execute work items in parallel as long as those work items does
  not relate to each other by a parent (including grand-parents) - child relationship. i.e.
  work items in $FWD^1 - SCC^1$ and $FWD^2$ cannot be executed in parallel. Therefore
  $FWD^1 - SCC^1$ and $FWD^2$ must belong to two different partitions in the
  induced equivalence relation (Partition 1 and Partition 2 in Figure~\ref{fig:dcscc-proof}).
  But we should be able to execute $FWD^1 - SCC^1$ and
  $BWD^1 - SCC^1$ in parallel as they dont relate to each other with a (grand)parent-child
  relationship. In other words, $FWD^1 - SCC^1$ and $BWD^1 - SCC^1$ belong to the same partition.
  Similarly $FWD^2$ and $BWD^1 - SCC^1$ should belong to the same partition
  as they also can be executed in parallel. This shows that two partitions
  (partition belong to $FWD^1 - SCC^1$ and partition belong to $FWD^2$) has an
  intersection (See Figure~\ref{fig:dcscc-proof}). This also shows that the induced
  relation does not divide work items into partitions $\implies$ The induced relation is
  not an equivalence relation $\implies$ R is not a strict weak ordering relation.

  Therefore, our original assumption is wrong, i.e. there cannot be a strict weak
  ordering relation to divide work items in DCSCC and therefore we cannot express
  DCSCC in an AGM.
\end{proof}

\begin{figure}
\centering
\includegraphics[width=.55\linewidth]{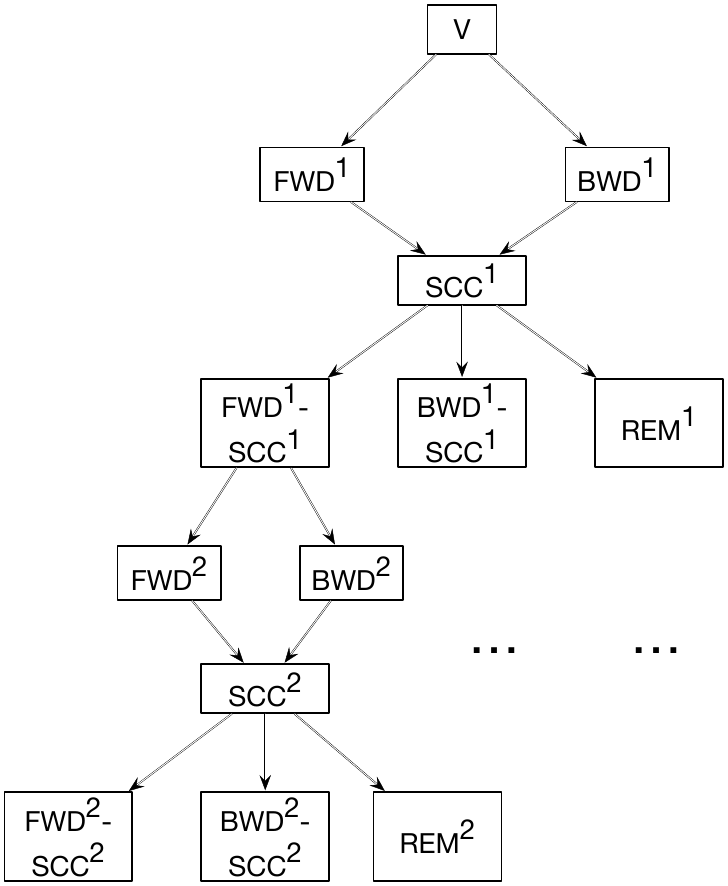}
\caption{Parallel work in DCSCC algorithm.}
\vspace{1ex}
\label{fig:dcscc}
\end{figure}

\begin{figure}
\centering
\includegraphics[width=.55\linewidth]{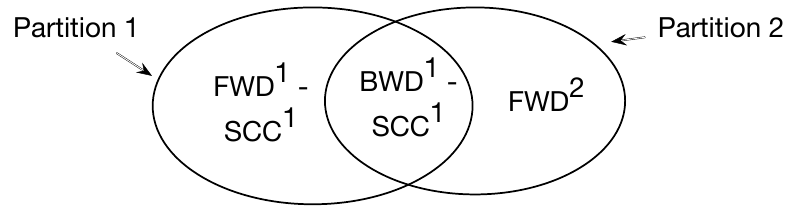}
\caption{$BWD^1 - SCC^1$ is the common intersection for two partitions.}
\vspace{1ex}
\label{fig:dcscc-proof}
\end{figure}

%\subsection{References}
%TODO
% Appendix B, and does not continue the present hierarchy
\end{document}